\documentclass[12pt,reqno]{article}

\usepackage[usenames]{color}
\usepackage{amssymb}
\usepackage{graphicx}
\usepackage{amscd}
\usepackage{cyrillic}

\usepackage[colorlinks=true,
linkcolor=webgreen,
filecolor=webbrown,
citecolor=webgreen]{hyperref}

\definecolor{webgreen}{rgb}{0,.5,0}
\definecolor{webbrown}{rgb}{.6,0,0}

\usepackage{breakurl}
\usepackage{color}
\usepackage{fullpage}
\usepackage{float}

\usepackage{graphics,amsmath,amssymb}
\usepackage{amsthm}
\usepackage{amsfonts}
\usepackage{latexsym}
\usepackage{epsf}

\newcommand{\seqnum}[1]{\href{https://oeis.org/#1}{\underline{#1}}}
\def\modd#1 #2{#1\ \mbox{\rm (mod}\ #2\mbox{\rm )}}

\def\shuff{\, \sha \,}
\DeclareMathOperator{\epp}{epp}
\DeclareMathOperator{\opp}{opp}

\title{Borders, Palindrome Prefixes, and Square Prefixes}

\author{Daniel Gabric and Jeffrey Shallit\footnote{Research supported by NSERC under grant 2018-04118.} \\
School of Computer Science \\
University of Waterloo \\
Waterloo, ON  N2L 3G1 \\
Canada\\
{\tt dgabric@uwaterloo.ca} \\
{\tt shallit@uwaterloo.ca} }

\begin{document}

\maketitle

\theoremstyle{plain}
\newtheorem{theorem}{Theorem}
\newtheorem{corollary}[theorem]{Corollary}
\newtheorem{lemma}[theorem]{Lemma}
\newtheorem{proposition}[theorem]{Proposition}

\theoremstyle{definition}
\newtheorem{definition}[theorem]{Definition}
\newtheorem{example}[theorem]{Example}
\newtheorem{conjecture}[theorem]{Conjecture}

\theoremstyle{remark}
\newtheorem{remark}[theorem]{Remark}

\begin{abstract}
We show that the number of length-$n$ words over a $k$-letter alphabet
having no even palindromic prefix is the same as the number of length-$n$
unbordered words, by constructing an explicit bijection between
the two sets.  A slightly different but analogous
result holds for those words having no odd palindromic prefix.
Using known
results on borders, we get an asymptotic enumeration for the 
number of words
having no even (resp., odd) palindromic prefix .
We obtain an analogous result for words having no nontrivial
palindromic prefix.
Finally, we obtain
similar results for words having no square prefix, thus proving a 
2013 conjecture of Chaffin, Linderman, Sloane, and Wilks.
\end{abstract}

\section{Introduction}

In this note, we work with finite words over a finite alphabet $\Sigma$.
For reasons that will be clear later, we assume without loss of generality
that $\Sigma = \Sigma_k = \{ 0,1,\ldots, k-1 \}$ for some integer $k \geq 1$.

We index words starting at position $1$, so that
$w[1]$ denotes the first symbol of $w$, and $w[i..j]$ is the factor
beginning at position $i$ and ending at position $j$.

We let $w^R$ denote the {\it reverse} of a word; thus, for example,
$ ({\tt drawer})^R = \tt reward$.   A word $w$ is a {\it palindrome}
if $w = w^R$; an example in English is the word
{\tt radar}.   A palindrome is
{\it even} if it is of even length, and {\it odd} otherwise.  
If a palindrome is of length $n$, then its {\it order} is defined
to be $\lfloor n/2 \rfloor$.  A palindrome is {\it trivial} if it is
of length $\leq 1$, and {\it nontrivial} otherwise.  We are not interested
in trivial palindromic prefixes in this paper, since every nonempty word
has such prefixes of length $0$ and $1$.

A word $w$ has an {\it even palindromic prefix} 
(resp., {\it odd palindromic prefix}) if there is
some nonempty prefix $p$ (possibly equal to $w$) that is a palindrome of
even (resp., odd) length.  Thus, for example, the English word 
{\tt diffident} has the even palindromic prefix {\tt diffid} of order $3$,
and the English word {\tt selfless} has an odd palindromic prefix
{\tt selfles} of order $3$.

A {\it border} of a word $w$ is a word $u$, $0 < |u| < |w|$,
that is both a prefix and
suffix of $w$.  If a word has a border, we say it is {\it bordered}, and
otherwise it is {\it unbordered}.  For example, {\tt alfalfa} is
bordered, but {\tt chickpea} is unbordered.
Bordered and unbordered words have been studied for almost fifty years;
see, for example, \cite{Silberger:1971,Ehrenfeucht&Silberger:1979}.

Call a border $u$ of a word $w$ {\it long} if $|u|> |w|/2$ and
{\it short} otherwise.  
If a word has a long border $u$, then by considering the overlap of the
two occurrences of $u$, one as prefix and one as suffix, we see that
$w$ also
has a short border.  Given a word $w$, its set of short border
lengths is 
$\{ 1 \leq i \leq |w|/2 \ : \ w[1..i] \text{ is a border of } w\}$.

By explicit counting for small $n$,
one quickly arrives at the conjecture that $u_k (n)$,
the {\it number\/}
of length-$n$ words over $\Sigma_k$
that are unbordered,
equals $v_k (n)$, the {\it number\/}
of length-$n$ words over $\Sigma_k$ having no
even palindromic prefix.   This seems to be true,
despite the fact that the individual words being
counted differ in the two cases.  As an example
consider $0011$, which has an even palindromic prefix $00$ but is
unbordered.    Similarly, if $t_k (n)$ denotes the number of
length-$n$ words over $\Sigma_k$ having no nontrivial
odd palindromic prefix,
it is natural to conjecture that $t_k (n) = u_k (n)$ for $n$ odd, and
$t_k (n) = k u_k (n-1)$ for $n$ even.  The first few terms of
the sequences $(t_2(n))$ and $(v_2 (n))$ are given in the following
table.

\begin{table}[H]
\begin{center}
\begin{tabular}{c|cccccccccccc}
$n$ & 1 & 2 & 3 & 4 & 5 & 6 & 7 & 8 & 9 & 10 & 11 & 12 \\
\hline
$t_2 (n)$ &  2&   4&   4&   8&  12&  24&  40&  80& 148& 296& 568&1136\\
$v_2 (n) = u_2(n)$ &   2&   2&   4&   6&  12&  20&  40&  74& 148& 284& 568&1116\\
\end{tabular}
\end{center}
\end{table}
The sequence $t_2 (n)$ is sequence
\seqnum{A308528} in the On-Line Encyclopedia of
Integer Sequences (OEIS) \cite{oeis}, and 
the sequence $u_2 (n)$ is sequence
\seqnum{A003000}.

In fact, even more seems to be true:  if $S$ is
any set of positive integers,
then the number of length-$n$ words $w$
for which $S$ specifies the lengths of all the short borders of $w$ is 
exactly the same as the number of length-$n$ words having even palindromic
prefixes with orders given by $S$.   
A similar, but slightly different claim seems to hold
for the odd palindromic
prefixes.  How can we explain this?

The obvious attempts at a bijection
(e.g., map $uvu$ to $u u^R v$) don't work, because (for example)
$00100$ and $00010$ both map to $00001$.  Nevertheless, there is a
bijection, as we will see below, and this bijection provides even more
information.

\section{A bijection on $\Sigma^n$}
\label{secf}

The {\it perfect shuffle} of two words $x$ and $y$, both
of length $n$, is written
$x \shuff y$, and is defined as follows:  if $x = a_1 a_2 \cdots a_n$
and $y = b_1 b_2 \cdots b_n$, then
$$ {x \shuff y} = a_1 b_1 a_2 b_2 \cdots a_n b_n.$$
Thus, for example, ${\tt clip} \shuff {\tt aloe}$ = {\tt calliope}.

Clearly $(x \shuff y)^R = y^R \shuff x^R$, a fact we use below.

\begin{lemma}
Let $x$ be an even-length word, and (uniquely)
write $x = y \shuff z$ for words $y, z$ with $|y| = |z|$.
Then $x$ is a palindrome iff $z = y^R$.
\label{one}
\end{lemma}

\begin{proof}
Suppose $x$ is palindrome.  Then $x = t t^R$ for some word $t$.

By ``unshuffling'', write
$t$ as $(t_1 \shuff t_2) a$, for words $t_1$ and $t_2$,
where $a$ is either empty or
a single letter, depending on whether $|t|$ is even or odd.
Then 
$$x = t t^R =  (t_1 \shuff t_2) a a (t_1 \shuff t_2)^R =
(t_1 \shuff t_2)aa (t_2^R \shuff t_1^R) = (t_1 a t_2^R)\shuff (t_2 a t_1^R).$$
It follows that $y = t_1 a t_2^R$ and $z = t_2 a t_1^R$,
and hence $z = y^R$.

Similarly, suppose $z = y^R$.   Write
$y = y_1 a y_2$, where $y_1, y_2$ are words of equal length,
and $a$ is either empty or a single letter, depending on whether
$|y|$ is even or odd.   Then $z = y_2^R a y_1^R$.   Hence 
$$x = (y_1 a y_2) \shuff (y_2^R a y_1^R) =
(y_1 \shuff y_2^R) aa (y_2 \shuff y_1^R).$$
Letting $t = (y_1 \shuff y_2^R)a$, we  see that
$x = t t^R$, and so $x$ is a palindrome.
\end{proof}

For a related result, see \cite{Rampersad&Shallit&Wang:2011}.

We now define a certain map from $\Sigma^n$ to $\Sigma^n$, as follows:
$$ f(x) := (y \shuff z^R) a$$
if $x = yaz$ with $|y| = |z|$ and $a$ empty or a single letter
(depending on whether $|x|$ is even or odd).
Thus, for example, $f({\tt preserve}) = {\tt perverse}$
and $f({\tt cider}) = {\tt cried}$.   Clearly this map is a bijection.

\begin{theorem}
Let $w$ be a word and let $1 \leq i \leq |w|/2$.
Then $w$ has a border of length $i$ iff
$f(w)$ has an even palindromic prefix of order $i$.
\label{two}
\end{theorem}

Roughly speaking, this theorem says that $f$ ``maps borders to orders''.

\begin{proof}
Suppose $w$ has a border of length $i$.  Then $w = uvu$, where
$|u| = i$.   Write $v = v_1 a v_2$, where $|v_1| = |v_2|$ and $a$
is either empty, or a single letter, depending on whether $|v|$
is even or odd.  Then $$f(w) = f(u v_1 a v_2 u) = ((u v_1)\shuff(v_2 u)^R ) a
= (u \shuff u^R) (v_1 \shuff v_2^R) a,$$
which by Lemma~\ref{one} has a palindromic prefix of length $2i$ and
order $i$.

Suppose $f(w)$ has an even palindromic prefix of order $i$.
Write $w = yaz$, so that $f(w) = (y \shuff z^R) a$.   
Write $y = y_1 y_2$ and $z = z_1 z_2$ such that
$|y_1| = |z_2| = i$.
Now $$f(w) = ((y_1 y_2) \shuff (z_2^R z_1^R)) a =
(y_1 \shuff z_2^R) (y_2 \shuff z_1^R) a .$$
It follows that $y_1 \shuff z_2^R$ is a palindrome and
thus $z_2 = y_1$ by Lemma~\ref{one}.
Hence $w = yaz = y_1 y_2 a z_1 z_2 = y_1 y_2 a z_1 y_1$
has a length-$i$ border, namely $y_1$.  
\end{proof}

\begin{corollary}
Let $S \subseteq \{1, \ldots, \lfloor n/2 \rfloor\}$.  Then
the number of length-$n$ words whose short borders
are exactly those in $S$ equals the number of length-$n$
words whose even palindromic prefixes are of orders
exactly those in $S$.

In particular, this holds for $S = \emptyset$, so
the equality $u_k (n) = v_k (n)$ holds for all $k \geq 1$ and $n \geq 1$:
the
number of length-$n$ words that are unbordered is the same as
the number of length-$n$ words having no even palindromic prefix.
\end{corollary}

\begin{proof}
As we have seen in Theorem~\ref{two},
the map $f$ is a bijection from the
the first set to the second.
\end{proof}

\begin{example}
As an example, consider the length-$8$ binary words
with short borders of length $1$ and $3$ only.  There
are 8 of them:
$$ \{ 01000010,  01001010,  01010010,
01011010,  10100101,  10101101,  10110101,  10111101 \}.$$
By applying the map $f$ to each
word, we get the length-$8$ binary
words having even palindromic prefixes of orders $1$ and
$3$ only:
$$ \{ 00110000,  00110001,  00110010,  00110011,
11001100,  11001101,  11001110,  11001111 \}.$$
\end{example}

Let $\epp_{k,S} (n)$ denote the number of length-$n$ words over
a $k$-letter alphabet having even palindromic prefixes of
order $i$ for each $i \in S$, and no other orders.

\begin{proposition}
We have $\epp_{k,S} (n+1) = k \cdot \epp_{k,S} (n)$ for $n$ even.
\label{evenprop}
\end{proposition}

\begin{proof}
Let $n$ be even.
Let $w$ be a word over a $k$-letter alphabet with even palindromic
prefix orders given by $S$, and let $a$ be a single letter.
Then clearly $wa$ has exactly the same palindromic prefixes as $w$.
Since $a$ is arbitrary, the result follows.
\end{proof}

\section{Odd palindromic prefixes}

Let $S$ be any subset of $\{ 1,2,\ldots, \lfloor n/2 \rfloor \}$.
Let $\opp_{k,S} (n)$ denote the number of length-$n$ words over
a $k$-letter alphabet having odd palindromic prefixes of
order $i$ for each $i \in S$, and no others.  Notice that we are not
concerned here with those words having the trivial odd palindromic prefix
of a single letter.
\begin{proposition}
We have $\opp_{k,S} (n+1) = k \cdot \opp_{k,S} (n)$ for $n$ odd.
\label{prop5}
\end{proposition}

\begin{proof}
Exactly like the proof of Proposition~\ref{evenprop}.
\end{proof}

\begin{theorem}
We have
\begin{itemize}
\item[(a)] $\opp_{k,S} (n) = \epp_{k,S} (n)$ for $n$ odd; and
\item[(b)] $\opp_{k,S} (n) = k \cdot \epp_{k,S} (n-1)$ for $n$ even.
\end{itemize}
\end{theorem}

\begin{proof}
We begin by proving 
$\opp_{k,S} (n) = k \cdot \epp_{k,S} (n-1)$ for $n$ odd.   We do this
by creating a $k$ to $1$ map from the length-$n$ words with 
odd palindromic
prefix orders given by $S$ to the length-$(n-1)$ words with even
palindromic prefix orders given by $S$.

Here is the map.  Let $w = a_1 a_2 \cdots a_n$ be a word of odd length,
and define $g(w) = (a_1 + a_2) (a_2 + a_3) \cdots (a_{n-1} + a_n)$, where
the addition is performed modulo $k$.  We claim that this is a $k$ to $1$
map, and furthermore, it maps words with odd palindromic prefix orders
given by $S$ to words with even palindromic prefix orders also given by $S$.

To see the first claim, observe that if both $g(w)$ and $a_1$ are given,
then we can uniquely reconstruct $w$.  Since $a_1$ is arbitrary, this
gives a $k$ to $1$ map.  

To see the second claim, suppose $w = a_1 a_2 \cdots a_n$
has an odd palindromic prefix of
order $i$.   Then $a_{2i+2-j} = a_j$ for $1 \leq j \leq i+1$.
Hence, applying the map $g$ to a prefix of $w$ we get
\begin{align*}
 & g(a_1 a_2 \cdots a_{2i} a_{2i+1} )  \\
\quad &= (a_1 + a_2) (a_2+a_3) \cdots (a_{2i} + a_{2i+1})   \\
\quad &= (a_1+a_2) (a_2+a_3) \cdots (a_{i-1} + a_i) (a_i + a_{i+1}) 
(a_{i+1} + a_i) (a_i + a_{i-1}) \cdots (a_3 + a_2) (a_2 + a_1),
\end{align*}
which is clearly an even palindrome of order $i$.

On the other hand, if 
$(a_1+a_2) (a_2+a_3) \cdots (a_{2i-1} + a_{2i}) (a_{2i} + a_{2i+1})$
is a palindrome, then by examining the two elements in the middle,
we get $a_{i} + a_{i+1} \equiv \modd{a_{i+1} + a_{i+2}} {k}$,
which forces $a_i = a_{i+2}$.  Continuing from the middle out to
the end, we successively obtain
$a_{i-1} = a_{i+3}$, \ldots, $a_1 = a_{2i+1}$,
which shows that $w$ starts with an odd palindrome of order $i$.

Hence for $n$ odd we get
\begin{equation}
\opp_{k,S} (n) = k \cdot \epp_{k,S} (n-1) =  \epp_{k,S} (n) ,
\label{opp1}
\end{equation}
where we have used Proposition~\ref{evenprop}.

For $n$ even we get
\begin{align*}
\opp_{k,S} (n) &= k \cdot \opp_{k,S} (n-1) \quad \text{(by Proposition~\ref{prop5})} \\
& = k \cdot \epp_{k,S} (n-1) 
	\quad \text{ (by Eq.~\eqref{opp1})},
\end{align*}
which completes the proof.
\end{proof}

\begin{corollary}
Consider words over a $k$-letter alphabet.

For $n$ odd, we have $t_k(n) = u_k(n)$; that is,
the number of length-$n$ words having no nontrivial
odd palindromic prefix is the same as the number
of length-$n$ words that are unbordered.

For $n$ even, we have $t_k(n) = k u_k(n-1)$; that is,
the number of length-$n$ words having no nontrivial odd
palindromic prefix is $k$ times the number of length-$(n-1)$
unbordered words.
\end{corollary}

\begin{remark}
It is seductive, but wrong, to think that the map $g$ also maps
even-length palindromic prefixes in a $k$ to $1$ manner to odd-length
palindromic prefixes, but this is not true (consider what happens to the
center letter).
\end{remark}

\section{An application}

As an application of our results
we can (for example)
determine the asymptotic fraction 
of length-$n$ words having no nontrivial even
palindromic prefix (resp., having no nontrivial odd palindromic prefix).

\begin{corollary}
For all $k \geq 2$ there is a constant $\gamma_k > 0$ such that
the number of length-$n$ words having no nontrivial even
palindromic prefix (resp., having no nontrivial odd palindromic prefix)
is asymptotically equal to $\gamma_k \cdot k^n$.
\end{corollary}

\begin{proof}
Follows immediately from the same result for unbordered words; 
see \cite{Nielsen:1973,Blom:1995,Holub&Shallit:2016}.  For related results,
see \cite{Richmond&Shallit:2014}.
\end{proof}

\section{Interlude:  the permutation defined by $f$}

The map $f$ defined in Section~\ref{secf} can be considered
as a permutation on $a_1 a_2 \cdots a_n$.  In this case,
we write it as $f_n$.  For example, if
$n = 7$, the resulting permutation $f_n$ is
$$ \left( \begin{array}{ccccccc}
	1 & 2 & 3 & 4 & 5 & 6 & 7 \\
	1 & 7 & 2 & 6 & 3 & 5 & 4 
	\end{array} \right) .$$
This is an interesting permutation that has been well-studied in the
context of card-shuffling, where it is called the {\it milk shuffle}.
A classic result about the milk shuffle is the following \cite{Levy:1951}:
\begin{theorem}
The order of the permutation $f_n$ is the least $m$ such that
$2^m = \pm 1 $ {\rm (mod} $2n-1${\rm )}.
\end{theorem}
This is sequence \seqnum{A003558} in the OEIS.

\section{No palindromic prefix}

In this section we consider the words having no nontrivial palindromic prefix.
(Recall that a palindrome is trivial if it is of length $\leq 1$.)
This is only of interest for alphabet size $k \geq 3$, for if $k = 2$,
the only such words are of the form $01^i$ and $1 0^i$.

Let $A_k (n)$ denote the number of such words over a $k$-letter alphabet.
We use the technique of \cite{Blom:1995,Holub&Shallit:2016} to show
that $A_k(n) \sim \rho_k k^n$ for a constant $\rho_k$ and large $n$.
First we need a lemma, which can essentially be found in (for example)
\cite[Prop.~6]{Blondin-Masse&Brlek&Garon&Labbe:2011}.

\begin{lemma}
Let $w$ be a palindrome and let $p$ be a proper
palindromic prefix of $w$.
If $|p| > |w|/2$, then $w$ also has a nonempty palindromic prefix of length
$< |w|/2$.
\label{blond}
\end{lemma}

\begin{proof}
If $p$ is a prefix of $w$, then $p^R$ is a suffix of $w^R$.
Since both $p$ and $w$ are palindromes, this means $p$ is a suffix of $w$.
Hence there exist nonempty words  $y, z$ such that $w = py = zp$.
By the Lyndon-Sch\"utzenberger theorem \cite{Lyndon&Schutzenberger:1962}
there exist $u, v$ with $u$ nonempty, and an integer
$e \geq 0$ such that 
$z = uv$, $p = (uv)^e u$, and $y = vu$.   Since $|p| > |w|/2$, it 
follows that $|u| \leq |y|< |w|/2 < |p|$.
Since $p$ is a palindrome, we have
$w = zp = zp^R$.   Since $w$ is a palindrome, we have
$py = w = w^R = p z^R$.   So $vu = y = z^R = v^R u^R$, and so
$u = u^R$.   Thus $u$ is a nonempty palindromic prefix of  $z$,
which is a prefix of $w$, and $|u| < |w|/2$.
\end{proof}

\begin{lemma}
Let $w, a$ be words with $|a| \leq 1$.  Then $wa$ has a nontrivial
palindromic prefix iff $waw^R$ has a nontrivial proper palindromic prefix.
\label{waw}
\end{lemma}

\begin{proof}
One direction is trivial.  

For the other direction, let the shortest nontrivial proper
palindromic prefix of $s := waw^R$ be $q$.   If $|q| \leq |wa|$, then
$q$ is a prefix of $wa$ as desired.  
Otherwise we have $|q| > |wa| \geq |s|/2$.  Then by Lemma~\ref{blond},
the word $s$ also has a nontrivial palindromic prefix of
length $< |s|/2$, contradicting the definition of $q$.
\end{proof}

\begin{proposition}
For $n \geq 1$ we have
\begin{align}
A_k(2n) &= k A_k(2n-1) - A_k(n)  \label{p1} \\
A_k(2n+1) &= k A_k(2n) - A_k(n+1)  \label{p2} .
\end{align}
\end{proposition}

\begin{proof}
Consider the words of length $2n-1$ having no nontrivial palindromic prefix.
By appending a new letter, we get $k A_k(2n-1)$ words.  However, 
some of these words can be palindromes of length $2n$, and we do not
want to count these.  By Lemma~\ref{waw}, the number of length-$2n$
palindromes having no proper palindromic prefix is $A_k (n)$.
This gives \eqref{p1}.

A similar argument works to prove \eqref{p2}.
\end{proof}

For $k = 3$ the corresponding sequence is given below and is sequence
\seqnum{A252696} in the OEIS:
\begin{table}[H]
\begin{center}
\begin{tabular}{c|cccccccccccc}
$n$ & 1 & 2 & 3 & 4 & 5 & 6 & 7 & 8 & 9 & 10 & 11 & 12 \\
\hline
$a_3 (n)$ &    3&     6&    12&    30&    78&   222&   636&  1878&  5556& 16590& 49548&148422\\
\end{tabular}
\end{center}
\end{table}

Now define $T_k (n)$ by $T_k(n) = A_k(n) k^{-n}$.  From
\eqref{p1} and \eqref{p2} we get
\begin{align*}
T_k(2n) &= T_k(2n-1) - T_k(n) k^{-n} \\
T_k(2n+1) &= T_k(2n) - T_k(n+1) k^{-n} .
\end{align*}
It now follows that
\begin{equation}
T_k(2n) = T_k(2n-2) - (k+1) T_k(n) k^{-n} . \label{p3}
\end{equation}
By telescoping cancellation applied to \eqref{p3}, we now get
$$ T_k(2n) = {{k-1} \over k} - (k+1) \sum_{i=2}^n T_k(i) k^{-i} ,$$
or
$$ T_k(2n) = 2 - (k+1) \sum_{i=1}^n T_k(i) k^{-i} .$$

Next, define $h(X) = \sum_{n \geq 1} T_k(n) X^n$.
Since $0 \leq T_k (n) \leq 1$ for all $n$, it follows that
the series defining $h(X)$ is convergent for $|X| < 1$.  Then
$$ \rho_k := \lim_{n \rightarrow \infty} T_k(n) = 2 - (k+1) h(1/k) ,$$
where $\rho_k$ is the limiting frequency of words having no
nontrivial palindromic prefix.

Using \eqref{p3}, we now get
\begin{align*}
h(X) (1-X) &= T_k(1) X + \sum_{n \geq 2} (T_k(n)-T_k(n-1)) X^n \\
&= T_k(1) X + \left( \sum_{n \geq 1} (T_k(2n)-T_k(2n-1)) X^{2n} \right) +
	\left( \sum_{n \geq 1} (T_k(2n+1) - T_k(2n)) X^{2n+1} \right) \\
&= X - \left( \sum_{n \geq 1} T_k(n) k^{-n} X^{2n} \right) -
	\left( \sum_{n \geq 1} T_k(n+1) k^{-n} X^{2n+1} \right) \\
&= X - h(X^2/k) - {k \over X} \left(f (X^2/k) - {{X^2}\over k}\right) \\
&= 2X - \left( 1 + {k \over X} \right) h(X^2/k), \\
\end{align*}
and so we get a functional equation for $h(X)$:
$$ h(X) = {{2X} \over {1-X}} + {{X+k} \over {X(X-1)}} h(X^2/k).$$

By iterating this functional equation, and using the fact
that $h(\epsilon) \sim \epsilon$ for small real $\epsilon$,
we get an expression for
$h(1/k)$:
$$ \left( \lim_{n \rightarrow \infty} {{ \prod_{i=1}^n (k^{2^i} + 1) } \over
	{k^{n+1} \prod_{i=1}^n (k^{2^i - 1} - 1) }} \right)
- 2 \sum_{n \geq 1} k^{2^{2n-1} - 2n} (k^{2^{2n-1} - 1} + 1) 
	{{ \prod_{i=1}^{2n-2}  (k^{2^i} + 1) } \over
	{ \prod_{i=1}^{2n}  (k^{2^i - 1} - 1) }} .
$$
This is very rapidly converging; for $k = 3$ only $6$ terms are enough
to get 60 decimal places of $h(1/k)$:
\begin{align*}
h(1/3) &= 0.430377520029471213293382335121830467895548542549528870740458 \cdots \\
\rho_3 &= 0.27848991988211514682647065951267812841780582980188451703816 \cdots
\end{align*}

\section{Square prefixes}

It is natural to conjecture that our bijections connecting
words with no border and no
even palindromic prefix might also apply to words having no square
prefix.  However, this is not the case.  Let
$s_k (n)$ denote the number of length-$n$ words over 
$\Sigma_k$ having no square prefix.
When $k = 2$,
for example, the two sequences $s_k (n)$ and $v_k(n)$ differ
for the first time at $n = 10$, as the following table indicates.
\begin{table}[H]
\begin{center}
\begin{tabular}{c|cccccccccccc}
$n$ & 1 & 2 & 3 & 4 & 5 & 6 & 7 & 8 & 9 & 10 & 11 & 12 \\
\hline
$v_2 (n)$ &   2&   2&   4&   6&  12&  20&  40&  74& 148& 284& 568&1116\\
$s_2 (n)$ &  2&   2&   4&   6&  12&  20&  40&  74& 148& 286& 572&1124\\
\end{tabular}
\end{center}
\end{table}
\noindent The sequence $s_2(n)$ is sequence
\seqnum{A122536} in the OEIS.

Chaffin, Linderman, Sloane, and Wilks
\cite[\S 3.7]{Chaffin&Linderman&Sloane&Wilks:2013}
conjectured that $s_2 (n) \sim \alpha_2 \cdot 2^n$ for a constant
$\alpha_2 \doteq 0.27$.  In this section we prove this conjecture in more
generality.

\begin{theorem}
The limit $\lim_{n \rightarrow \infty} s_k (n)/k^n$ exists and equals
a constant $\alpha_k$ with $\alpha_k > 1-1/(k-1)$.
\end{theorem}

\begin{proof}
Let $d_k(n) = k^n - s_k(n)$ be the number of length-$n$ words over
$\Sigma_k$ having a nonempty square
prefix, and let
$c_k(n)$ be the number of squares of length $2n$ over $\Sigma_k$
having no nonempty
proper square prefix.  
Hence $c_k(1) = k$ and $c_k(2) = k(k-1)$.  

The first few values of $c_2 (n)$ and $d_2 (n)$ are given in the
following table.
\begin{table}[H]
\begin{center}
\begin{tabular}{c|cccccccccccc}
$n$ & 1 & 2 & 3 & 4 & 5 & 6 & 7 & 8 & 9 & 10 & 11 & 12 \\
\hline
$c_2 (n)$ &    2&   2&   4&   6&  10&  20&  36&  72& 142& 280& 560&1114\\
$d_2 (n)$ &   0&   2&   4&  10&  20&  44&  88& 182& 364& 738&1476&2972\\
\end{tabular}
\end{center}
\end{table}
\noindent The sequence $c_2 (n)$ is sequence \seqnum{A216958} in the OEIS, and
the sequence $d_2 (n)$ is sequence \seqnum{A121880}.

Let $w$ be a word of length $n$.  Either its shortest square prefix is of
length 2 (and there are $c_k(1) k^{n-2}$ such words), or of length 4
(and there are $c_k(2) k^{n-4}$ such words), and so forth.

So $d_k(n)$, the number of words of length $n$ having a nonempty
square prefix,
is exactly
$\sum_{2i \leq n} c_k(i) k^{n-2i}$.
Hence $d_k(n)/k^n = \sum_{2i \leq n} c_k(i) k^{-2i}$.  Thus
$\lim_{n \rightarrow \infty} d_k(n)/k^n$ exists
iff the infinite sum
$\sum_{i = 1}^\infty c_k(i) k^{-2i}$ converges.
But, since $c_k (i) \leq k^i$, this sum converges to some constant
$\beta_k< 1/(k-1) $, by comparison with
the sum $\sum_{i=1}^\infty k^{-i} = 1/(k-1)$.    
It follows that $s_k (n) = k^n - d_k (n) \sim (1-\beta_k) k^n$.
Letting $\alpha_k = 1-\beta_k$, the result follows.
\end{proof}

To estimate the value of $\beta_k$ (and hence $\alpha_k$) we use the
inequalities
$$ \sum_{i=1}^n c_k (i) k^{-2i} \leq \beta_k \leq
\left( \sum_{i=1}^n c_k (i) k^{-2i} \right) + \sum_{i=n+1}^\infty k^{-i}
= \left( \sum_{i=1}^n c_k (i) k^{-2i} \right) + (k^{-n})/(k-1) .$$
For example, if we take $n = 20$, we get
$\beta_2 \in (0.7299563,0.7299574)$ and hence
$\alpha_2 \in (0.2700426, 0.2700437)$.
This can be compared to the analogous constant
$\gamma_2 \doteq 0.2677868$ for even palindromes.

\subsection*{Acknowledgments}

We are grateful to the referees for their helpful comments.

\newcommand{\noopsort}[1]{} \newcommand{\singleletter}[1]{#1}

\end{document}